\documentclass[runningheads]{llncs}
\newif\ifpublish
\newif\ifanonymous
\publishtrue   
\anonymousfalse

\ifanonymous
  
  \author{Anonymous Submission}
  \authorrunning{Anonymous}
  \institute{}
\else
  \author{Rithwik Kerur\inst{1}\orcidID{0009-0000-4834-9445} \and
          Dahlia Malkhi\inst{1}\orcidID{0000-0002-7038-7250} \and
          Michael K. Reiter\inst{2}\orcidID{0000-0001-7007-8274}}
  \authorrunning{R. Kerur et al.}
  \institute{University of California, Santa Barbara, Santa Barbara, CA, USA\\
    \email{\{rkerur,dahliamalkhi\}@ucsb.edu}
    \and
    Duke University, Durham, NC, USA\\
    \email{michael.reiter@duke.edu}}
\fi

\date{September 2026}
\titlerunning{eAVID: AVID with Post-Dissemination Pruning}
\usepackage{graphicx} 

\usepackage{amssymb}   
\usepackage{xcolor}
\usepackage{xspace} 
\usepackage{comment}
\usepackage{amsmath}
\usepackage[normalem]{ulem}
\usepackage{subcaption}
\usepackage{wrapfig}
\usepackage{algorithm}
\usepackage[inline]{enumitem}
\usepackage{algpseudocode}
\usepackage{comment, booktabs}
\newcommand{\sys}{eAVID\xspace}
\newcommand{\dispersedsimplex}{DispersedSimplex\xspace}
\newcommand{\sysConsensus}{eAVID-Simplex\xspace}

\ifpublish
    \newcommand{\DM}[1]{}
    \newcommand{\MKR}[1]{}
    \newcommand{\TODO}[1]{}
    \newcommand{\Rithwik}[1]{}
\else
    \newcommand{\DM}[1]{{\color{blue} DM: #1}}
    \newcommand{\MKR}[1]{{\color{purple} MKR: #1}}
    \newcommand{\TODO}[1]{{\color{red} TODO: #1}}
    \newcommand{\Rithwik}[1]{{\color{red} Rithwik: #1}}
\fi

\title{\sys: Asynchronous Verifiable Information Dispersal with Post-Dissemination Pruning}

\begin{document}

\maketitle
\begin{abstract}
\DM{An alternative abstract:}
\MKR{The file contains a commented-out, slightly tighter abstract that I slightly prefer.}
The well-known Asynchronous Verifiable Information Dispersal (AVID) problem lets a sender disperse a message across $N=3F+1$ nodes such that it remains recoverable despite up to $F$ Byzantine failures and unbounded message delays. An optimal AVID scheme requires $3\times$ the message size in communication and storage: up to $F$ nodes may be delayed in responding, and among the remaining $2F+1$ respondents, up to $F$ may be Byzantine.

This paper presents eAVID, an elastic AVID scheme that provides two key improvements. First, after dissemination, nodes can prune up to half of their stored information without requiring anyone to reconstruct or recode the message.
Second, pruning is enabled by an extremely simple protocol. Nodes collect acknowledgments from one another confirming receipt of their coded pieces, after which each node locally prunes its stored information.

eAVID achieves this $2\times$ reduction while making two tradeoffs: the sender generates $2N$ fragments rather than $N$, and pruning may necessitate contacting $2F+1$ nodes for message retrieval rather than $F+1$.

\sys was implemented in
DispersedSimplex, a simple-to-understand BFT consensus protocol that erasure-codes its blocks. The implementation demonstrates two features. First, \sys achieves these improvements using a flat
erasure-coding scheme that requires no metadata or bookkeeping at the nodes
during reconstruction. Second, it delivers the storage savings without a
performance cost: with all nodes responsive, pruning fires on over $99\%$ of
blocks and steady-state per-node storage falls by $42$--$46\%$ for committees
of $10$ to $22$ nodes. Throughput and latency match the unmodified protocol showing that we can achieve these storage savings without any performance costs. 

\end{abstract}

\section{Introduction}

Distributed fault-tolerant algorithms are often designed for a double worst case. First, they wait for only $N-F$ nodes to participate, since the network may indefinitely delay messages from up to $F$ nodes. Second, they assume that up to $F$ of the responding nodes may themselves be faulty. As a result, these protocols typically incur suboptimal communication and storage costs.

Concretely, suppose a sender $S$ wishes to store a block $B$ across a system of $N$ nodes such that $B$ remains recoverable despite up to $F$ failures while minimizing communication and storage. This is known as the \emph{information dispersal} (IDA) problem~\cite{Rabin}. In Byzantine settings, where both the sender and storage nodes may behave arbitrarily, additional verification is required to guarantee unique recoverability. This strengthened problem is known as \emph{verifiable information dispersal} (VID)~\cite{GGJR}.

A classical approach to VID combines erasure coding with cryptographic verification. The sender encodes $B$ using a $(k,k+m)$ erasure code and distributes the resulting $k+m$ fragments, one per node (with certain information added for verification), such that any $k$ fragments suffice to reconstruct $B$. In a synchronous system, where message delivery is guaranteed within a known bound $\Delta$, the sender can use a $(k=2F+1,k+m=N)$ encoding and distribute one fragment to every node within $\Delta$. This incurs a communication overhead of $3/2 \times |B|$, while the total storage consumed across all $N$ nodes is likewise $3/2 \times |B|$, which is optimal for $F$-resilience. Indeed, if exactly $F$ nodes fail, the surviving nodes collectively store exactly $|B|$ data.
In contrast, \emph{asynchronous} VID protocols (AVID)~\cite{AVID,DispersedLedger,PignoletSPAA26} resort to a $(F+1,N)$ encoding. Consequently, the sender incurs a communication overhead of $3 \times |B|$, the system stores $3 \times |B|$, and after $F$ failures the surviving nodes still store $2 \times |B|$.

There is a seeming tradeoff: either rely on a strong synchrony assumption for optimal dispersal and storage complexity, or incur a $2 \times$ bandwidth and storage increase in practical settings. The natural question is whether it has to be one or the other. 

In the crash-fault setting, this question was answered by eAID~\cite{eaid},
which introduced \emph{post-dissemination pruning}: nodes reclaim storage
autonomously once dissemination completes, with no re-encoding and no
coordination. Byzantine faults make the same idea considerably harder. A
crash-fault node can take an acknowledgement at face value, but a Byzantine
sender may equivocate about which fragments it dispersed, and faulty nodes may
fabricate acknowledgements for fragments they never stored. A node that prunes
on such evidence could destroy the last copy of data it believes is safely
replicated elsewhere. Extending pruning to this setting therefore requires that
a node's decision rest on verifiable evidence about what was dispersed and on
proof that enough correct nodes actually hold it.

In this paper, we address this question with \sys, an elastic AVID protocol
where elasticity manifests itself in two complementary ways:

\begin{enumerate}
    \item \sys uses a $(2F+1,2N)$ erasure code, rather than either $(F+1, N)$ or a $(2F+1,N)$ encoding. A sender  disperses the resulting $2N$ fragments in an \emph{optimistically informed elastic} manner. Specifically, the sender can flexibly assign one or more fragments to each node, subject only to the requirement that at least $2F+1$ fragments survive despite up to $F$ failures.

    \item Nodes also manage their local storage elastically. Each node independently detects redundant fragments based on responses from other nodes and may safely discard a portion of its locally stored fragments.
\end{enumerate}

For example, the sender may initially distribute one fragment to each of the $N$ nodes, achieving the optimal communication overhead of $3/2 \times |B|$, while the system stores the same optimal amount. If exactly $F$ nodes fail, the surviving nodes collectively store exactly $|B|$. The downside is that dispersal can complete only when $N$ nodes have responded. To tolerate network delays, the sender may instead initially distribute two fragments to each node and wait for only $N-F$ responses. Later, once a node learns that responses have been received from all $N$ nodes, it may safely discard one of its two local fragments. Other combinations of fragments and responses are possible, including assigning varying amounts of fragments to different nodes, or staggering the dispersal one fragment at a time based on responses. Thus, \sys enables varying the sender's communication overhead between $3/2 \times |B|$ and $3 \times |B|$, while storage is adjusted post-dissemination automatically to the optimal amount.

Crucially, the scheme preserves a simple and uniform reconstruction procedure: any $2F+1$ fragments can recover each data block, requiring neither special bookkeeping nor a specialized reconstruction procedure.


\section{Background}
\label{sec:Background}

\subsection{Verifiable Information Dispersal}

The information dispersal problem, originally formulated in Rabin's seminal
IDA method~\cite{Rabin}, considers a sender that holds a message $M$ of size
$|M|$ and wishes to spread it across a set of $N$ nodes while providing both
data availability and storage efficiency. Rabin's formulation targets benign
(crash) faults: the dispersal must guarantee that, despite up to $F < N/2$
nodes crashing, $M$ can be reconstructed from the fragments held by the
surviving nodes.

\emph{Verifiable Information Dispersal} (VID), introduced by Garay, Gennaro,
Jutla, and Rabin~\cite{GGJR}, extends IDA to a setting in which the nodes, and
even the sender, may be Byzantine. Cachin and Tessaro~\cite{AVID} later
extended VID to the asynchronous network model (AVID), and a long line of
subsequent work~\cite{HGR,DispersedLedger,DumboMVBA,LDDRVXG21} refined its
storage and communication trade-offs.

\medskip

\textbf{The reconstruction challenge:}
Byzantine faults make reconstruction a two-sided problem. A corrupt
\emph{node} might serve a fabricated fragment during retrieval, and a corrupt
\emph{sender} might disseminate inconsistent fragments to begin with.
Reconstruction must therefore tolerate bad fragments from (i)~up to $F$ corrupt
nodes, and (ii)~honest nodes holding fragments planted by a corrupt sender.
Absent any safeguard, decoding different subsets of $F+1$ fragments could yield
a different message each time, breaking the consistency that essentially every
use of VID depends on.

\medskip

\textbf{Merkle-tree authentication:}
\sys authenticates fragments with a \textbf{Merkle tree}, one
of three standard \textit{fragment authentication} mechanisms; the other two
(hash vectors and polynomial/homomorphic commitments) are surveyed in
Appendix ~\ref{app:related-work}. Like those mechanisms, it provides a binding
guarantee: the sender is pinned to a single value and cannot open it to two
different values after the fact.
The sender disperses only the root $r_M$,
and each fragment $f_i$ travels with a logarithmic-size proof $\pi_i$ attesting
that $f_i$ is the $i$-th leaf under $r_M$. Each node stores only its fragment,
an $O(\log N)$ proof, and the root.

\medskip

\textbf{Codeword Consistency:} 
Authentication certifies only that a fragment is the one the sender
\emph{committed to}, not that the committed fragments lie on a single codeword. A corrupt sender might send inconsistent fragments. Several cryptographic methods allow to bind fragments to a committed reconstruction result, e.g., homomorphic commitments or fingerprints, also described in Appendix~\ref{app:related-work}.
Another option, \emph{detect-and-reject}, catches an inconsistent sender at reconstruction. This method is used in \sys and is described next.

\medskip

\textbf{Post-dissemination detect-and-reject:}
Upon completion of the dissemination, a candidate $M$ is decoded, the Merkle tree recomputed
from the re-encoded fragments, and rejected on any mismatch with
the committed root. Since the verdict is fixed once the root is fixed at
dispersal, every honest retriever reaches the same outcome regardless of which
fragments it used, restoring consistency even against a corrupt
sender~\cite{ByzantineErasureCodedStorage,DispersedLedger}. This approach provides succinct $O(\log N)$ proofs and an $O(1)$ root, and minimal assumptions (only a collision resistant hash is required). It forgoes only recovery of
data from a \emph{corrupt} sender which VID does not need. Note that while we mention Merkle trees, this strategy could also be used with hash vectors or polynomial commitments which would have different sized proofs/roots.

\textbf{Fragments, codes, and storage blowup:}
Storage efficiency comes from erasure coding. An erasure code is
parameterized by $(k, n)$: it encodes $M$ into $n$ \emph{fragments} such
that any $k$ suffice to reconstruct $M$. With a Reed-Solomon
code~\cite{Reed-Solomon}, each fragment has size $|M|/k$, total storage is
$\frac{|M|}{k}\times n$, and the storage blowup over $|M|$ is
$\frac{n}{k}$. In the standard asynchronous Byzantine parameterization, one
fragment is stored per node, so $n = N$, with $N = 3F+1$ and $k = F+1$.
Each node stores a fragment of size $|M|/(F+1)$ and the blowup
is $\frac{N}{F+1} \approx 3$ in the optimal-resilience regime
$F < N/3$. Dispersal completes once $N-F$ nodes have responded, while up to
$F$ corrupt nodes may contribute fabricated or mismatched fragments during
retrieval.

\subsection{Merkle-Based AVID Solutions}
\label{sec:background-avid}
We now describe the Merkle-based AVID schemes on which \sys
builds; the remaining, non-Merkle solutions are surveyed in
Section~\ref{app:related-work}.

The original AVID scheme of Cachin and Tessaro~\cite{AVID} integrates an
asynchronous reliable broadcast with erasure coding under optimal Byzantine
resilience $F < N/3$. The sender encodes $M$ using a $(k, N)$ erasure code
and computes a vector $D = [H(f_1), \ldots, H(f_N)]$ of per-fragment hashes;
receivers run a Bracha-style three-phase exchange~\cite{Bracha}
(\textsf{send}/\textsf{echo}/\textsf{ready}) carrying fragments and the hash
vector until $k+F$ matching \textsf{ready} messages on a common $D$ are
observed. AVID-H replaces the hash vector with a Merkle tree over the
fragments, so each receiver stores only its fragment, a logarithmic-size
proof, and the root: yielding communication blowup $O(N)$ and
storage blowup approaching $\frac{N}{F+1}$ (bounded by $3 + o(1)$).
AVID uses the hash vector $D$ for authentication while AVID-H uses the Merkle
Root. Neither adds a dedicated corrupt-sender remedy, instead leaning on the
inner reliable broadcast to agree on a single $D$/root before any fragment is
accepted.

AVID-H has two structural costs that motivate later work. First, every
dispersal pays for an $O(N^2)$-message reliable broadcast on full fragments, even when the network is healthy. Several later works aim to optimize the echoed payload. Second, the encoding parameters are fixed at $k = F+1$:
nodes always store a fragment of size $|M|/k$, and there is no mechanism to
shed storage when more than $k$ honest nodes successfully complete the
dispersal.

The same paper's AVID-RBC variant addresses the storage cost by layering an
additional round of erasure coding on top of AVID-H's reliable broadcast,
driving storage blowup down toward the optimal $\frac{N}{N-F}$ at the cost of
encoding the message \emph{twice}~\cite{AVID}. Concretely, the broadcast uses
an $(N-2F, N)$ code, and after the message has been committed, nodes
re-encode it with an $(N-F, N)$ code for long-term storage. The reduction is
only achieved by a coordinated re-encode, not by anything a node can do on its
own.

\section{System Model}
\label{sec:problem}
 
\subsection{Parties and network.}
The system consists of $N$ parties, $P_1, \ldots, P_N$, connected by
pairwise authenticated point-to-point channels. At most $F < N/3$
parties may be Byzantine and behave maliciously. The remaining parties are
\textit{correct} and follow the protocol. Any party may act as a
\textit{sender} that initiates a dispersal of a message $M$. The
sender itself may be Byzantine: a key objective of our protocol is
that safety holds against a faulty sender (e.g., one that
equivocates), and that honest parties detect equivocation and reject
inconsistent dispersals. We assume static membership: the set of
parties is fixed once the protocol begins.
 
We adopt the classical \textit{asynchronous} model such that the
adversary may delay messages between honest parties arbitrarily but cannot drop them.
 
\subsection{Cryptographic assumptions.}
We assume a computationally bounded adversary and the following
standard primitives.
\begin{itemize}
    \item \textbf{Authenticated Channels.} We assume authenticated
    point-to-point links between every pair of parties. If a correct
    party $j$ receives a message $M$ on the link from party $i$, then
    $i$ indeed sent $M$ to $j$. The adversary may delay 
    messages arbitrarily, but cannot forge, modify, or replay messages
    on links between correct parties, nor impersonate a correct
    sender. Messages between correct parties are eventually delivered.
    \item \textbf{Collision-resistant hash function.} We assume a
    cryptographic hash function $H$ for which it is computationally
    infeasible to find $x \neq x'$ with $H(x) = H(x')$. Merkle trees
    instantiated with $H$ are binding: given a root $r$, no
    polynomial-time adversary can produce two distinct valid
    openings to the same leaf position.
\end{itemize}

\subsection{Verifiable Information Dispersal.}
\sys realizes a Verifiable Information Dispersal (VID) abstraction
with two protocols, \textsc{Disperse} and \textsc{Retrieve}, that
satisfy the following properties.
\begin{itemize}
    \item \textbf{Binding.} Once the first honest party completes
    \textsc{Disperse} for $\mathit{id}$, there is a unique value
    $v^\star \in \{\textsc{Completed}(M^\star), \textsc{Abort}\}$ such that every
    honest party that completes \textsc{Disperse} for $\mathit{id}$ outputs
    $v^\star$. If $v^\star = \textsc{Completed}(M^\star)$, then $M^\star$ is a
    unique committed message and every honest party that completes
    \textsc{Retrieve} on $\mathit{id}$ outputs $M^\star$.

    \item \textbf{Validity.} If the sender is honest with input $M$, then
    $v^\star = \textsc{Completed}(M^\star)$ with $M^\star = M$.

    \item \textbf{Agreement.} Any two honest parties that complete
    \textsc{Retrieve} on the same $\mathit{id}$ output the same value.

    \item \textbf{Termination of dispersal.} If the sender is honest with
    input $M$, then every honest party eventually completes \textsc{Disperse}
    for $\mathit{id}$ with output $\textsc{Completed}(M)$.

    \item \textbf{Totality.} If any honest party completes \textsc{Disperse}
    for $\mathit{id}$ with a \textsc{Completed} output, then every honest party
    eventually completes \textsc{Disperse} for $\mathit{id}$.

    \item \textbf{Termination of retrieval.} If some honest party has
    completed \textsc{Disperse} for $\mathit{id}$ with a \textsc{Completed} output,
    then every honest party that invokes \textsc{Retrieve} on $\mathit{id}$
    eventually completes it with output $M^\star$.
\end{itemize}

\textbf{Performance measures.}
With these assumptions as our starting point, we shift our focus to the metric we aim to optimize. Our primary goal is to enhance performance by optimizing  \textbf{storage cost}.
We define \textbf{storage cost} as the aggregate size of storage used across all nodes to persist the fragments associated with a message. For  a baseline,  full replication---where a message of size $B$ is sent to all $N$ nodes---incurs a total storage cost of $O(N \cdot B)$. We also define \textbf{per-node storage} as the storage required on a responsive node in the system.
\section{Solution}
\label{sec:Solution}

\subsection{Dissemination and Reconstruction}
\label{sec:dissemination-protocol}

Algorithm~\ref{alg:ida-party} formalizes the protocol. The sender encodes $M$
into $2N$ fragments with a $(2F{+}1, 2N)$ Reed-Solomon code, builds a Merkle
tree over the fragments, and sends each node $k$ its assigned pair $S_k$
together with the root $r_M$ and Merkle proofs
(\textsc{Disperse}, Algorithm~\ref{alg:ida-party}, lines~1--5). Since $k = 2F{+}1$
rather than $F{+}1$, fragments are half the size, so sending two uses roughly the same
bandwidth as one larger fragment under a standard $(F{+}1, N)$ scheme. The
sender could also vary the number of fragments per node. We reserve this
adaptive strategy for future work.

Upon receiving \textsc{Disperse}, a node verifies its Merkle proofs and
broadcasts an \textsc{Echo} carrying its fragments
(Algorithm~\ref{alg:ida-party}, lines~6--8). A correct node echoes at most once, committing
to a single root. A node that did not receive \textsc{Disperse} can still
participate: once it has $F{+}1$ \textsc{Echo}s for a common root, it decodes,
re-encodes its own pair, and echoes (Algorithm~\ref{alg:ida-party}, lines~14--15).

\begin{algorithm}[!htbp]
\caption{\sys Party}
\label{alg:ida-party}
\scriptsize
\begin{algorithmic}[1]
\Require local id $j$
\Statex \textbf{State per $\mathit{id}$:} root $r \gets \bot$;\ \
        $\mathit{Ech}[\cdot]$ (root $\mapsto$ echoers, init $\emptyset$);\ \
        $\mathit{src}[\cdot]$ (party $\mapsto$ root, init $\bot$);\ \
        buffer $B$; fragments $\mathcal{F}$;\ \
        $\mathit{Dn}[\cdot]$ (party $\mapsto$ root, init $\bot$);\ \
        flags $\mathit{echoed}, \mathit{doneSent}, \mathit{decided}, \mathit{pruned} \gets \mathit{false}$

\Function{Disperse}{$M$} \Comment{sender only}
    \State $(f_1, \ldots, f_{2N}) \gets \Call{Encode}{M, 2F{+}1, 2N}$;\ \
           $r_M, \{\pi_i\}_{i=1}^{2N} \gets \Call{BuildMerkleTree}{f_1, \ldots, f_{2N}}$
    \State \textbf{for each} $k \in \{1,\ldots,N\}$, where $S_k \gets \{2k{-}1, 2k\}$:
    \State \quad \textbf{send} $\langle \textsc{Disperse}, \mathit{id}, r_M, \{(i, f_i, \pi_i) : i \in S_k\} \rangle$ to $k$
\EndFunction

\Statex
\State \textbf{upon} $\langle \textsc{Disperse}, \mathit{id}, r_M, X \rangle$ from sender, \textbf{if not} $\mathit{echoed}$:
\State \quad $P \gets \{(i, f_i, \pi_i) \in X : \Call{VerifyMerkle}{r_M, i, f_i, \pi_i} \textbf{ and } i \in S_j\}$
\State \quad \textbf{if} $P \neq \emptyset$ \textbf{then} $\mathit{echoed} \gets \mathit{true}$;\ \
       $\mathcal{F} \gets P$;\ \
       broadcast $\langle \textsc{Echo}, \mathit{id}, r_M, P \rangle$

\Statex
\State \textbf{upon} $\langle \textsc{Echo}, \mathit{id}, r_M, P \rangle$ from $k \notin \mathit{Ech}[r_M]$:
\State \quad \textbf{if} $\mathit{src}[k] \notin \{\bot, r_M\}$ \textbf{then return}
       \Comment{one echo per sender}
\State \quad $P' \gets \{(i, f_i, \pi_i) \in P : \Call{VerifyMerkle}{r_M, i, f_i, \pi_i} \textbf{ and } i \in S_k\}$
\State \quad \textbf{if} $P' = \emptyset$ \textbf{then return} \Comment{echo must carry $k$'s own verified fragment}
\State \quad $\mathit{src}[k] \gets r_M$;\ \ $\mathit{Ech}[r_M] \gets \mathit{Ech}[r_M] \cup \{k\}$;\ \
       \textbf{if} $r_M = r$ \textbf{then} $B \gets B \cup P'$;\ \ \Call{TryFinishFrom}{}
\State \quad \textbf{if not} $\mathit{echoed}$ \textbf{and} $|\mathit{Ech}[r_M]| \geq F+1$ \textbf{and} fragments echoed under $r_M$ decode and re-encode to root $r_M$ \textbf{then}
\State \quad \quad $P_j \gets \text{this party's pair } S_j \text{ from re-encoding}$;\ \
           $\mathit{echoed} \gets \mathit{true}$;\ \ $\mathcal{F} \gets P_j$;\ \ broadcast $\langle \textsc{Echo}, \mathit{id}, r_M, P_j \rangle$
\State \quad \textbf{if} $r = \bot$ \textbf{and} $|\mathit{Ech}[r_M]| \geq 2F+1$ \textbf{then}
       \Comment{adopt the root that reached the echo quorum}
\State \quad \quad $r \gets r_M$;\ \ $B \gets \{(i,f_i,\pi_i) \text{ echoed under } r_M : \Call{VerifyMerkle}{r,i,f_i,\pi_i}\}$;\ \
       \Call{TryFinishFrom}{};\ \ \Call{CheckDone}{}
\State \quad \textbf{if} $r_M = r$ \textbf{and} $|\mathit{Ech}[r]| = N$ \textbf{and} $\mathit{doneSent}$ \textbf{and not} $\mathit{decided}$ \textbf{then}
       \Comment{fast path}
\State \quad \quad $\mathit{decided} \gets \mathit{true}$;\ \ \textbf{output} $\langle \mathit{id}, \textsc{out}, \textsc{Completed} \rangle$

\Statex
\State \textbf{upon} $\langle \textsc{Done}, \mathit{id}, r_k \rangle$ from $k$ with $\mathit{Dn}[k] = \bot$:
       \Comment{one \textsc{Done} per sender}
\State \quad $\mathit{Dn}[k] \gets r_k$;\ \ \textbf{if} $r = \bot$ \textbf{and} $|\{k' : \mathit{Dn}[k'] = r_k\}| \geq F+1$ \textbf{then}
       $r \gets r_k$;\ \ \textbf{load} $B$ with fragments echoed under $r_k$ that verify against $r$
\State \quad \Call{TryFinishFrom}{};\ \ \Call{CheckDone}{}

\Statex
\Function{CheckDone}{}
    \State \textbf{if} $r = \bot$ \textbf{then return};\ \ $D \gets \{k : \mathit{Dn}[k] = r\}$
    \State \textbf{if} $|D| \geq 2F + 1$ \textbf{and not} $\mathit{decided}$ \textbf{then}
           $\mathit{decided} \gets \mathit{true}$;\ \ \textbf{output} $\langle \mathit{id}, \textsc{out}, \textsc{Completed} \rangle$
    \State \textbf{if} $|D| = N$ \textbf{and} $\mathit{doneSent}$ \textbf{and not} $\mathit{pruned}$ \textbf{then}
           $\mathit{pruned} \gets \mathit{true}$;\ \ discard elements of $\mathcal{F}$ until $|\mathcal{F}| = 1$
\EndFunction

\Statex
\Function{TryFinishFrom}{} \Comment{Reconstruct $\to$ Done/Abort}
    \State \textbf{if} $r = \bot$ \textbf{or} $\mathit{doneSent}$ \textbf{or} $\mathit{decided}$ \textbf{then return}
    \State \textbf{if} $|\{i : (i,\cdot,\cdot) \in B\}| < 2F+1$ \textbf{then return}
    \State \textbf{if} $B$ decodes and re-encodes to a Merkle root equal to $r$ \textbf{then} \Comment{broadcast \textsc{Done}}
    \State \quad $\mathcal{F} \gets$ this party's pair $S_j$ from re-encoding
           \Comment{replaces the fragments stored at echo time}
    \State \quad $\mathit{doneSent} \gets \mathit{true}$;\ \ $\mathit{Dn}[j] \gets r$;\ \
       broadcast $\langle \textsc{Done}, \mathit{id}, r \rangle$
    \State \textbf{else} \Comment{non-codeword under the agreed root: abort locally}
    \State \quad $\mathit{decided} \gets \mathit{true}$;\ \ \textbf{output} $\langle \mathit{id}, \textsc{out}, \textsc{Abort} \rangle$
\EndFunction
\end{algorithmic}
\end{algorithm}

A node adopts a root once an echo quorum of $2F{+}1$ nodes supports it
(Algorithm~\ref{alg:ida-party}, lines~16--17). Any two such quorums intersect in a correct
node, which echoed only one root, so conflicting roots can never both be
adopted (Lemma~\ref{lem:root-agreement}).

After adopting a root, a node reconstructs via \Call{TryFinishFrom}{}
(Algorithm~\ref{alg:ida-party}, lines~28--36): if $\geq 2F{+}1$ verified fragments decode
and re-encode to the adopted root, it broadcasts \textsc{Done}. Otherwise it
outputs \textsc{Abort}. A node outputs \textsc{Completed} upon receiving
$2F{+}1$ \textsc{Done}s over its adopted root
(\Call{CheckDone}{}, Algorithm~\ref{alg:ida-party}, lines~23--25).

\sys also offers a \emph{fast path} to completion that runs concurrently with
the slow path: a node that has \textsc{Echo}s from all $N$ nodes over its
adopted root and has broadcast \textsc{Done} outputs \textsc{Completed}
immediately (Algorithm~\ref{alg:ida-party}, lines~18--19), saving a message delay. Both
predicates are checked as messages arrive, and the node outputs on whichever
fires first. The fast path is safe because of the $N$ echoers at least
$2F{+}1$ are correct and hold fragments at disjoint indices
(Lemma~\ref{lem:dispersal-safety}).

\begin{algorithm}[!htbp]
\caption{\sys Retrieval}
\label{alg:ida-retrieve}
\scriptsize
\begin{algorithmic}[1]
\Require adopted root $r^\star$ for $\mathit{id}$ (from Dispersal); retriever $R$
\Statex \textbf{State per $\mathit{id}$:} retrieval buffer $B_R$ (init $\emptyset$); flag $\mathit{retDone}$ (init $\mathit{false}$)

\Statex
\Function{Retrieve}{$\mathit{id}$}
    \State \textbf{if} $r^\star = \bot$ \textbf{then return} $\bot$ \Comment{no root adopted}
    \State broadcast $\langle \textsc{Get}, \mathit{id}, r^\star \rangle$ to all parties
\EndFunction

\Statex
\State \textbf{upon} $\langle \textsc{Get}, \mathit{id}, r' \rangle$ from $R$: \Comment{server $P_i$}
\State \quad $Y \gets \{(i, f_i, \pi_i) \in \mathcal{F} : \Call{VerifyMerkle}{r', i, f_i, \pi_i}\}$
       \Comment{stored fragments under $r'$, adopted or not}
\State \quad \textbf{if} $Y \neq \emptyset$ \textbf{then} send $\langle \textsc{GetResp}, \mathit{id}, Y \rangle$ to $R$
       \Comment{answer with whatever survives pruning}
\Statex
\State \textbf{upon} $\langle \textsc{GetResp}, \mathit{id}, Y \rangle$ from $k$: \Comment{retriever $R$}
\State \quad \textbf{if} $\mathit{retDone}$ \textbf{then return}
\State \quad $B_R \gets B_R \cup \{(i, f_i, \pi_i) \in Y : \Call{VerifyMerkle}{r^\star, i, f_i, \pi_i} \textbf{ and } i \in S_k\}$
       \Comment{verified under $r^\star$; index attributable to server $k$}
\State \quad \textbf{if} $|\{i : (i,\cdot,\cdot) \in B_R\}| \geq 2F+1$ \textbf{then}
\State \quad \quad $M \gets \Call{Decode}{B_R, 2F{+}1, 2N}$
\State \quad \quad $\mathit{retDone} \gets \mathit{true}$;\ \ \textbf{output} $\langle \mathit{id}, \textsc{ret}, M \rangle$
\end{algorithmic}
\end{algorithm}

\subsection{Post-Dissemination Pruning}
\label{sec:post-dissemination}
\sys's primary departure from prior Byzantine works is that every node can
prune a fragment post-dissemination from its own local view, without
re-encoding or requiring nodes to agree when to prune. After outputting \textsc{Completed},
each node continues collecting \textsc{Done}s in the background. Once a node
has \textsc{Done}s from all $N$ nodes and has itself broadcast \textsc{Done},
it discards one of its two fragments
(\Call{CheckDone}{}, Algorithm~\ref{alg:ida-party}, line~26), cutting its storage in
half. By Lemma~\ref{lem:pruning-safety}, this is safe regardless of which
other nodes have pruned: at least $2F{+}1$ correct nodes stored their pairs
before announcing, and each discards at most one, so $2F{+}1$ distinct
verified fragments always survive.

Through a combination of these approaches, \sys departs from previous work by:
\begin{enumerate*}
    \item Allowing different numbers of fragments to be stored by
    different nodes: so long as $2F+1$ fragments can be collected from
    correct nodes in the system, $M$ will remain available;
    \item Not requiring re-encoding of $M$ or re-construction of the
    Merkle tree;
    \item Not impacting the recovery computation, which always takes
    $2F+1$ verified fragments and reconstructs $M$ in the same manner,
    oblivious to the assignment of fragments to nodes;
    \item Not requiring maintaining meta-information per message $M$
    beyond the Merkle root.
\end{enumerate*}

\paragraph{Bandwidth and storage trade-off.}
Standard Asynchronous Byzantine dispersal approaches use an $(F+1, N)$ encoding scheme,
where each node stores a single fragment of size roughly
$\frac{|M|}{F+1}$. By using a $(2F+1, 2N)$ encoding scheme and sending
two fragments to each node, \sys maintains the \emph{same} aggregate
bandwidth utilization during dissemination: each fragment is of size
roughly $\frac{|M|}{2F+1}$, and each node receives two of them, so the
per-node payload matches the standard scheme. However, if the network is
healthy and
every node announces \textsc{Done} and those announcements reach a given
node, that node can subsequently drop one of
its two fragments and
retain only a single fragment of size $\frac{|M|}{2F+1}$, effectively
cutting steady-state storage in half compared to the standard
$(F+1, N)$ baseline. In adverse conditions, \sys gracefully degrades to
the standard storage footprint, ensuring that the optimistic case is
rewarded without penalizing the pessimistic one.

\subsection{Safety and Liveness}
\label{sec:proof}

\begin{lemma}
\label{lem:root-agreement}
No two correct parties adopt distinct roots for the same
$\mathit{id}$.
\end{lemma}
\begin{proof}[Proof sketch]
Any two echo quorums of size $2F{+}1$ intersect in $\geq F{+}1$ parties on
$N = 3F{+}1$, hence in at least one correct party, which echoes only one root.
So all echo-path adoptions agree on a single $r^\star$. A simple induction
shows every correct \textsc{Done} is over $r^\star$, so \textsc{Done}-path
adoptions also yield $r^\star$. The full proof appears in
Appendix~\ref{app:deferred-proofs}.
\end{proof}

\begin{lemma}[Root uniqueness]
\label{lem:root-uniqueness}
At most one root $r^\star$ is adopted by correct parties, and any
correct party that outputs does so with respect to $r^\star$.
\end{lemma}
\begin{proof}
Immediate from Lemma~\ref{lem:root-agreement}: a correct party
adopts $r^\star$, evaluates \Call{IsCodeword}{}, stores its assigned pair,
and only then broadcasts \textsc{Done} or outputs, and all correct adopters
share $r^\star$.
\end{proof}

\begin{lemma}[Binding]
\label{lem:binding}
Once the first correct party outputs for $\mathit{id}$, there is a unique value
$v^\star \in \{\textsc{Completed}(M^\star), \textsc{Abort}\}$ such
that any correct party that outputs for $\mathit{id}$ does so with
$v^\star$, and any correct party that completes \textsc{Retrieve} outputs
$M^\star$ (if $v^\star = \textsc{Completed}(M^\star)$).
\end{lemma}
\begin{proof}
By Lemma~\ref{lem:root-uniqueness} every correct party that outputs has adopted
the same $r^\star$, which commits a unique leaf
vector. The verdict \Call{IsCodeword}{} is a deterministic
function of $r^\star$, so all correct parties compute the same value.
Two cases:

\emph{Case A: $r^\star$ commits a valid $(2F{+}1,2N)$ codeword.}
Any $2F+1$ leaves verified under $r^\star$ interpolate the same degree-$\le 2F$
polynomial and re-encode to $r^\star$, so \Call{IsCodeword}{} returns true at
every correct party that reconstructs. No correct party takes the abort branch,
and each such party broadcasts \textsc{Done}
with value $r^\star$. A correct party outputs
$\textsc{Completed}(M^\star)$ upon receiving $2F+1$ \textsc{Done}s
over $r^\star$, or on the fast path upon accepting
\textsc{Echo}s over $r^\star$ from all $N$ parties after broadcasting
\textsc{Done} itself, where $M^\star$ is the
unique decoding. A correct \textsc{Retrieve} gathers $2F+1$ leaves under
$r^\star$, decodes, re-encodes (matching), and outputs $M^\star$. Thus
$v^\star = \textsc{Completed}(M^\star)$.

\emph{Case B: $r^\star$ commits a non-codeword.} No $2F+1$ leaves
verified under $r^\star$ re-encode to $r^\star$, so \Call{IsCodeword}{} returns
false at every correct party that reconstructs: no correct party
broadcasts \textsc{Done} over $r^\star$, and no $2F+1$
\textsc{Done}
quorum can form. No correct party outputs \textsc{Completed}.
The fast path requires $\mathit{doneSent}$, so it does not
fire either.
Every correct party
that adopts $r^\star$ and reconstructs outputs \textsc{Abort}. Thus
$v^\star = \textsc{Abort}$.
\end{proof}

\begin{lemma}[Termination of dispersal]
\label{lem:termination}
If the sender is correct, then every correct party eventually
outputs $\textsc{Completed}(M)$.
\end{lemma}
\begin{proof}
\textbf{Step 1 (echo and adopt).} The correct sender sends each
correct party its verified pair under $r_M$. Each such party echoes $r_M$ with
its fragment. There are $\geq 2F+1$ correct
parties, so every correct party eventually collects an echo quorum for $r_M$ and
adopts $r^\star = r_M$. No correct party echoes any other root (the sender sends
only $r_M$), so no other root is adopted.

\textbf{Step 2 (reconstruct and \textsc{Done}).} Each \textsc{Echo}
carries the echoer's fragment. Since the pairs $S_k$ are disjoint, the echoes of
$\geq F+1$ correct parties give every correct party $\geq 2F+2$ distinct verified
leaves in $B$. \Call{TryFinishFrom}{} finds \Call{IsCodeword}{} true, stores the party's pair, and
broadcasts \textsc{Done} with value $r_M$
to all parties.

\textbf{Step 3 (output).} Every correct party thus
broadcasts \textsc{Done}.
Correct parties broadcast to all, and messages between correct parties are
eventually received, so each correct party eventually receives these
$\geq 2F+1$ \textsc{Done}s over $r_M = r^\star$ and outputs
$\textsc{Completed}(M)$. No correct party aborts.
\end{proof}

\begin{lemma}[Validity]
\label{lem:validity}
If the sender is correct with input $M$, then every correct party
that outputs does so with $\textsc{Completed}(M)$.
\end{lemma}
\begin{proof}
By Lemma~\ref{lem:termination} every correct party outputs, and by
Lemma~\ref{lem:root-agreement} all adopt $r^\star = r_M$, which commits the
codeword of $M$. By Case A of Lemma~\ref{lem:binding}, the output is
$\textsc{Completed}(M^\star)$ with $M^\star = M$.
\end{proof}

\begin{lemma}[Agreement]
\label{lem:agreement}
Any two correct parties that complete \textsc{Retrieve} on the same
$\mathit{id}$ output the same value.
\end{lemma}
\begin{proof}
By Lemma~\ref{lem:binding} both adopt $r^\star$ and output the unique
$M^\star$ (or both are in the \textsc{Abort} case, in which
retrieval returns no value).
\end{proof}

\begin{lemma}[Reconstructability at completion]
\label{lem:dispersal-safety}
When the first correct party outputs $\textsc{Completed}(M^\star)$
for $\mathit{id}$, the fragments held by correct parties contain at least
$2F+1$ distinct verified fragments under $r^\star$, sufficient to reconstruct
$M^\star$.
\end{lemma}
\begin{proof}
Outputting \textsc{Completed} requires $2F+1$ \textsc{Done}s
received over
$r^\star$ from distinct senders, at least $F+1$ from correct
parties.
By authenticity of the channels a \textsc{Done} counted against a correct
party $k$ was genuinely sent by $k$, and the
$\mathit{Dn}[k] = \bot$ guard admits at most
one per sender, so
each of these $\geq F{+}1$ correct
parties executed the \textsc{Done} broadcast in
\Call{TryFinishFrom}{}.
Since
\Call{TryFinishFrom}{} stores the party's pair before
broadcasting, no correct party
broadcasts \textsc{Done} without
holding its verified pair, and before broadcasting it stored
its assigned pair.
Each correct party accepts only fragments assigned to itself, and the
partition is disjoint, so distinct correct parties hold disjoint fragments. The
$\geq F+1$ correct \textsc{Done}-broadcasters therefore hold
$\geq 2(F+1) = 2F+2$
distinct verified fragments, exceeding the $(2F{+}1)$ threshold.

If instead the first output was on the fast path, the
outputting party accepted \textsc{Echo}s over $r^\star$ from all $N$ parties,
at least $2F{+}1$ of them correct. A correct party echoes only fragments of
its own pair that verify under the root it echoes, and stores them before
echoing; thereafter $\mathcal{F}$ changes only when \Call{TryFinishFrom}{}
replaces it with the party's pair under its adopted root, which by
Lemma~\ref{lem:root-agreement} is $r^\star$, or when pruning reduces it to a
single fragment. Each of these $\geq 2F{+}1$ correct parties therefore holds
at least one fragment verified under $r^\star$, at distinct indices by
disjointness of the partition, meeting the threshold. This is the counting of
Lemma~\ref{lem:pruning-safety} with \textsc{Echo}s in place of
\textsc{Done}s.
\end{proof}

\begin{lemma}[Pruning safety]
\label{lem:pruning-safety}
If a correct party prunes a fragment for
$\mathit{id}$, then
$M^\star$ remains reconstructable
from the fragments held by correct parties.
\end{lemma}
\begin{proof}[Proof sketch]
A party prunes only after receiving \textsc{Done} from all $N$ parties, at
least $2F{+}1$ of which are correct and each stored its assigned pair before
broadcasting. Since pairs are disjoint and pruning discards at most one
fragment per party, $\geq 2F{+}1$ distinct verified fragments survive in every
mixed pruning state, meeting the $(2F{+}1,2N)$ decoding threshold. The full
proof appears in Appendix~\ref{app:deferred-proofs}.
\end{proof}

\begin{lemma}[Totality]
\label{lem:totality}
If any correct
party outputs $\textsc{Completed}(M^\star)$ for $\mathit{id}$, then every correct party
eventually outputs $\textsc{Completed}(M^\star)$. If any correct party outputs
$\textsc{Abort}$, then no correct party outputs $\textsc{Completed}$, and every correct
party that adopts $r^\star$ outputs $\textsc{Abort}$.
\end{lemma}
\begin{proof}[Proof sketch]
By Lemma~\ref{lem:root-agreement} all correct parties adopt the same
$r^\star$, and by Lemma~\ref{lem:binding} the verdict is deterministic.
If $r^\star$ commits a codeword, the $\geq F{+}1$ correct \textsc{Done}
broadcasters (or $\geq 2F{+}1$ correct echoers on the fast path) ensure every
correct party eventually adopts $r^\star$, collects $\geq 2F{+}1$ fragments
from disjoint pairs, reconstructs successfully, and broadcasts its own
\textsc{Done}, producing $\geq 2F{+}1$ \textsc{Done}s for every party.
If $r^\star$ commits a non-codeword, no correct party broadcasts \textsc{Done},
so no \textsc{Done} quorum forms and every adopter outputs \textsc{Abort}.
The full proof appears in Appendix~\ref{app:deferred-proofs}.
\end{proof}

\begin{lemma}[Retrieval termination]
\label{lem:ret-termination}
If dispersal for $\mathit{id}$
completed with a $\textsc{Completed}$ output, then every
correct party that invokes \textsc{Retrieve} eventually outputs $M^\star$.
\end{lemma}
\begin{proof}[Proof sketch]
By Lemma~\ref{lem:dispersal-safety}, at least $2F{+}1$ correct parties hold
verified fragments at distinct indices under $r^\star$ after completion,
whether or not any party has pruned (Lemma~\ref{lem:pruning-safety}). Each
replies to a \textsc{Get} request, so the retriever eventually accumulates
$\geq 2F{+}1$ fragments and decodes $M^\star$. The full proof appears in
Appendix~\ref{app:deferred-proofs}.
\end{proof}

We show that our protocol establishes the binding,
validity, agreement,
termination of dispersal, totality, and termination of retrieval as defined in
Section~\ref{sec:problem}.
Lemma~\ref{lem:pruning-safety} shows that these properties continue to hold
after post-dissemination pruning, in every state reachable by any subset of
correct parties having pruned.
\section{   Implementation}
\label{sec:implementation}

We implement \sys by building on an implementation~\cite{DispersedSimplexCode} of DispersedSimplex~\cite{DispersedSimplex}, a leader-based BFT protocol in which the leader of each round erasure-codes its block and disperses it across the committee. The protocol implementation is approximately 6,000 lines of Rust and uses a Merkle-based $(F{+}1, N)$ dispersal with one fragment per node making it a direct baseline against which to substitute \sys. We chose it as a starting point because it is simple to understand and already erasure-codes its payloads.

Our changes total approximately 200 lines. The encoding parameters become $(2F{+}1, 2N)$; each node is assigned the disjoint pair $S_k = {2k-1, 2k}$; and a node discards one fragment once it has observed a \textsc{Done} over its adopted root from all $N$ parties. We introduce no new message types and no additional rounds, and retrieval is unchanged. Each node persists a single RocksDB~\cite{RocksDB} record per block, holding its assigned fragments and their Merkle proofs. No node ever stores a complete block. We use ed25519~\cite{ed25519} signatures for single-signer identity and BLS~\cite{bls} signatures for aggregatable votes, so a quorum of \textsc{Echo} or \textsc{Done} messages verifies as one aggregate against a signer bitmap.

\subsection{Liveness in \dispersedsimplex}
\dispersedsimplex, like any deterministic consensus protocol, requires partial
synchrony for liveness, while its safety holds in asynchrony.
\dispersedsimplex's liveness argument requires the dispersal to complete
within a known number of network delays $\delta$ when the leader is correct.
In the implementation we build on, the slot timeout is
$\Delta_{\text{timeout}} \geq 4\delta$: honest nodes may enter a slot up to
$\delta$ apart, the leader's proposal takes $\delta$ to arrive, the
\textsc{Echo}s it triggers take another $\delta$, and the \textsc{Done}s that
follow a further $\delta$, after which every honest node Completes. A
node must not time out and complain before this point, or a correct leader's
slot would fail. Integration with \sys does not change any of the phases and thus maintains liveness under the same $\Delta_{\text{timeout}} \geq 4\delta$ bound.

\section{Evaluation}

Through our evaluation we aim to show that \sys matches the throughput and
latency of the dispersal each system uses today, while halving steady-state
per-node storage whenever the network is healthy enough for every node to
complete dispersal. Under adverse conditions, \sys falls back to the same
footprint as the baseline, so the optimistic case is rewarded without
penalizing the pessimistic one. We aim to do so by establishing the following
claims:

\begin{description}
  \item[\textbf{C1}.] In the common case, \sys maintains the same throughput
    and latency as \dispersedsimplex.
  \item[\textbf{C2}.] In the common case, \sys reduces per-node storage.
  \item[\textbf{C3}.] During adverse network conditions, \sys maintains the
    same per-node storage and performance as \dispersedsimplex.
\end{description}

We define per-node storage as the bytes a node retains per dispersed
block: its assigned fragments, their Merkle proofs, and the committed header.
Pruning reclaims only the fragment; the remaining metadata is a fixed cost per
block, which is why measured savings sit marginally below the code-rate bound.

\subsection{Experimental Setup}

All experiments run on a DigitalOcean Droplet provisioned with 32 vCPUs and
64\,GB of RAM. Our main configuration places the entire cluster on this single
machine, which lets us measure storage behavior without confounding it with
external network variance. Because storage reduction depends on the evolution of consensus state rather than on packet
delivery times, colocating replicas lets us reproduce partition scenarios
deterministically, free of the jitter that a wide-area deployment would
introduce. To confirm that these findings carry over to a real network, we
additionally run \sys across Droplets in
DigitalOcean's SFO3 region, driven by a colocated client that submits
transactions. We report those numbers in Section~\ref{sec:performance}.

\subsection{Performance}
\label{sec:performance}

\begin{figure}
    \centering
    \vspace{-2pt}
    \includegraphics[width=0.45\textwidth]{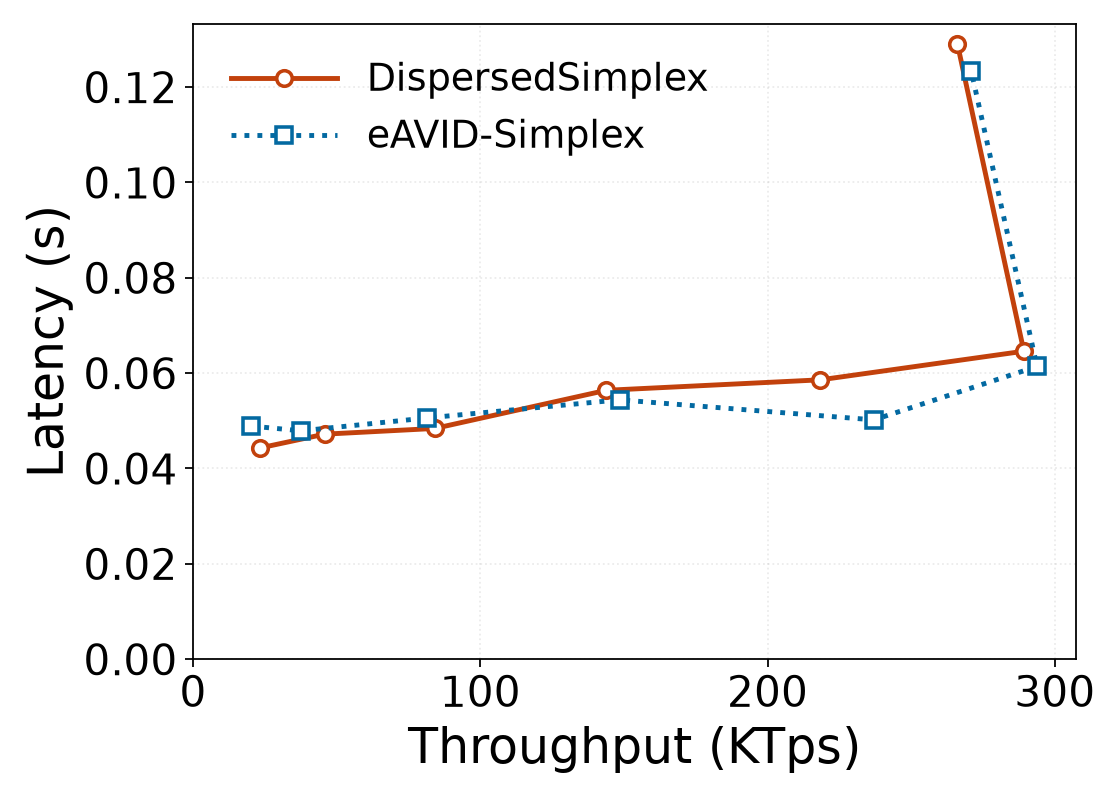}
    \vspace{-6pt}
    \caption{Dispersal Latency vs. throughput of DispersedSimplex and \sysConsensus with 10 nodes.}
    \label{fig:performance}
    \vspace{-10pt}
\end{figure}

Figure~\ref{fig:performance} plots latency against committed throughput for
DispersedSimplex with its native dispersal layer and for DispersedSimplex
instantiated with \sys, on ten nodes. The
leader generates the transactions when proposing, so block size sets the
offered load; we sweep it from $256$\,KiB to $16$\,MiB and run each
configuration for $60$ seconds. Latency is the mean time from a block's
proposal until a quorum of replicas has committed it. Both systems
trace the same curve: throughput grows almost linearly with block size while
latency stays between $44$ and $65$\,ms, and both saturate at an $8$\,MiB
block, at $289$\,KTps and $64.6$\,ms for the baseline against $294$\,KTps and
$61.4$\,ms for \sys. Across the sweep the two differ by
at most $8.4$\,ms, with neither ahead at every load; repeated runs of a single
configuration differ by $1.0$\,ms, so the differences below $4$\,MiB are at
the resolution of the measurement. \sys therefore matches the throughput and
latency of the dispersal layer it replaces, establishing C1.

\begin{table}[t]
  \caption{Storage \sys uses in excess of DispersedSimplex before pruning, as a
    percentage of the block. A node holds two fragments of $|M|/(2F{+}1)$ where
    the baseline holds one of $|M|/(F{+}1)$, an excess of
    $|M|/((2F{+}1)(F{+}1))$ together with a second Merkle path and one extra
    level on each path.}
  \label{tab:eavid-extra}
  \centering\small
  \setlength{\tabcolsep}{5pt}
  \begin{tabular}{rrrrrr}
    \toprule
    & & \multicolumn{4}{c}{Block size} \\
    \cmidrule(lr){3-6}
    $N$ & $F$ & 64\,KiB & 256\,KiB & 1\,MiB & 16\,MiB \\
    \midrule
    4   & 1  & 16.84\% & 16.71\% & 16.68\% & 16.67\% \\
    10  & 3  &  3.86\% &  3.66\% &  3.59\% &  3.57\% \\
    16  & 5  &  1.84\% &  1.61\% &  1.53\% &  1.52\% \\
    31  & 10 &  0.83\% &  0.54\% &  0.46\% &  0.43\% \\
    49  & 16 &  0.59\% &  0.32\% &  0.21\% &  0.18\% \\
    100 & 33 &  0.54\% &  0.18\% &  0.08\% &  0.05\% \\
    \bottomrule
  \end{tabular}
\end{table}

\subsection{Storage utilization}
\label{sec:eval-storage}

Figure~\ref{fig:storage-per-node} reports per-node storage for a node after a
given number of blocks for $N \in \{10, 16, 22\}$, under normal execution
(Figure~\ref{fig:storage-normal}) and with $F$ nodes unresponsive
(Figure~\ref{fig:storage-faults}).

Under normal execution, measured at $500$ blocks, \sys retains $76.8$\,MB
against \dispersedsimplex's $133.4$\,MB at $N = 10$ ($42.4\%$ less), $49.2$\,MB
against $89.0$\,MB at $N = 16$ ($44.8\%$), and $36.2$\,MB against $66.8$\,MB at
$N = 22$ ($45.9\%$). These figures track the analytic bound
$1 - (F{+}1)/(2F{+}1)$, which gives $42.9\%$, $45.5\%$, and $46.7\%$
respectively. Due to the fixed metadata, the total storage savings is not
exactly $50\%$, but trends towards it as we scale $N$ and the size of the block
as shown in Section~\ref{sec:overhead}. Pruning fired on $>99.0\%$ of stored blocks across the three configurations.

The saving is contingent on \textsc{Done} arriving from every node.
Figure~\ref{fig:storage-faults} repeats the experiment with $F$ nodes
unresponsive, leaving the live set at exactly the quorum $N - F$. The pruning
predicate requires $|D| = N$, so it never fires: across all three
configurations, $100\%$ of blocks retain both fragments. \dispersedsimplex is
unaffected, storing $1/(F{+}1)$ of each block whether or not nodes are down, so
the lines for it are identical in the two figures. \sys, holding two fragments
of $|M|/(2F{+}1)$ rather than one, instead stores $2(F{+}1)/(2F{+}1)$ of the
baseline: $152.6$\,MB against $133.4$\,MB at $N = 10$ ($14.4\%$ more),
$97.2$\,MB against $89.0$\,MB at $N = 16$ ($9.2\%$), and $71.4$\,MB against
$66.8$\,MB at $N = 22$ ($6.9\%$).

\begin{figure}
    \centering
    \vspace{-2pt}
    \begin{subfigure}{0.45\textwidth}
        \centering
        \includegraphics[width=\textwidth]{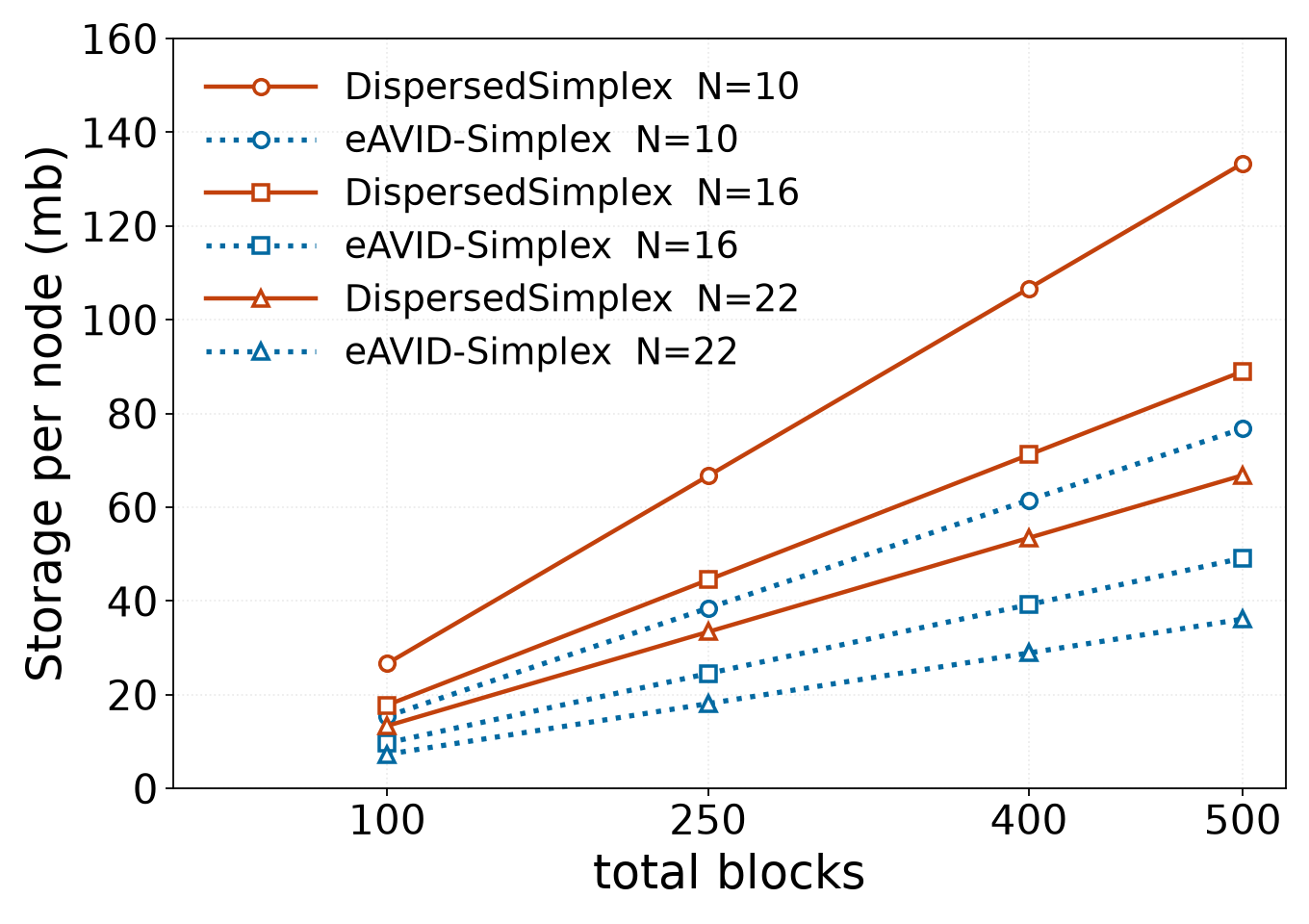}
        \caption{All nodes responsive.}
        \label{fig:storage-normal}
    \end{subfigure}
    \hfill
    \begin{subfigure}{0.45\textwidth}
        \centering
        \includegraphics[width=\textwidth]{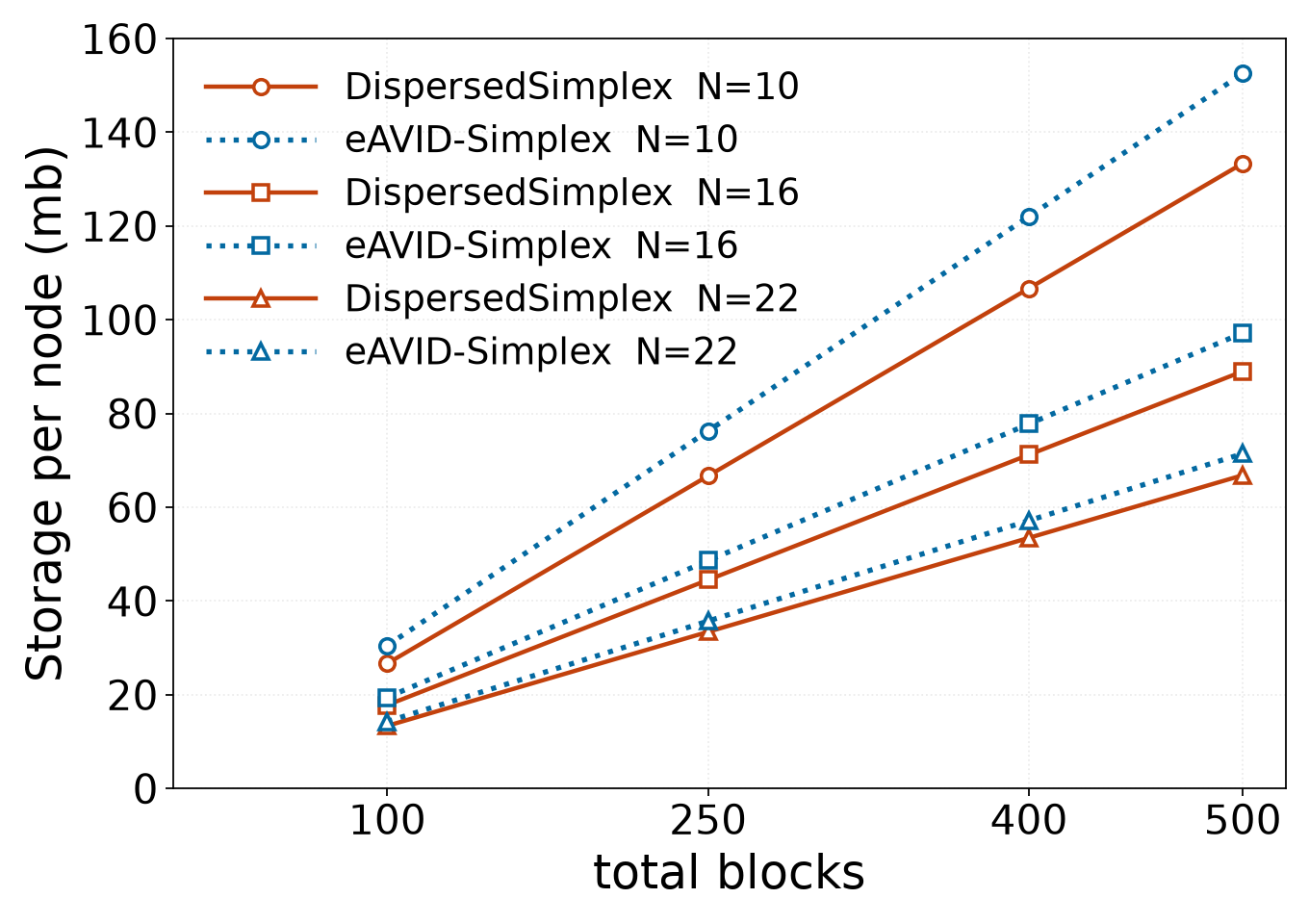}
        \caption{$F$ nodes unresponsive.}
        \label{fig:storage-faults}
    \end{subfigure}
    \vspace{-6pt}
    \caption{Experimental evaluation of storage utilization for \sys and
    \dispersedsimplex. With every node responsive
    (left), \sys prunes to a single fragment and stores less than
    \dispersedsimplex. With $F$ nodes unresponsive (right), \sys never prunes and retains both fragments.}
    \label{fig:storage-per-node}
    \vspace{-10pt}
\end{figure}

\subsection{Storage overhead}
\label{sec:overhead}

\sys's elasticity is not free, but the cost is transient.
Table~\ref{tab:eavid-extra} reports it as a fraction of the block \emph{before}
pruning. Because both schemes' fragments scale with $|M|$, the ratio is mostly
constant in block size and shrinks only with committee size, falling as
$O(1/F^{2})$ from $16.67\%$ at $N = 4$ to $1.52\%$ at $N = 16$ and $0.05\%$ at
$N = 100$. Once every node has completed dispersal the second proof is
discarded and this excess vanishes, leaving the node holding $45\%$ \emph{less}
than the baseline at $N = 16$ (Section~\ref{sec:eval-storage}). \sys therefore
spends $1.5\%$ of a block while dispersal is outstanding to reclaim $45\%$ of
one for the lifetime of the block.

\section{Conclusion}

We presented \sys, an asynchronous verifiable information dispersal protocol
that reduces per-node storage by nearly half in the common case through
post-dissemination pruning. Unlike prior approaches, \sys requires no
additional rounds, no re-encoding, and no coordination among nodes: each node
prunes independently based on its local view of completed dispersals. When all
nodes are responsive, \sys achieves up to $45\%$ storage savings over the
baseline while matching its throughput and latency. Under adverse conditions,
\sys gracefully falls back to the standard storage footprint, ensuring the
optimistic case is rewarded without penalizing the pessimistic one.

\newpage
\bibliographystyle{plain}
\bibliography{references}
\newpage
\appendix
\section{VID Landscape}
\label{app:related-work}
\subsection{Fragment Authentication Beyond Merkle Trees}
Besides the Merkle tree used by \sys, solutions typically authenticate fragments in
one of two other ways. As with the Merkle tree, each provides a binding
guarantee.

\paragraph{(i)~Hash vector.}
The sender broadcasts to nodes a vector
$D = [H({f}_1), \ldots, H({f}_N)]$ of per-fragment hashes, and a fragment is
accepted only if it matches its entry in $D$. This is the mechanism of the
original AVID~\cite{AVID}; it requires $O(N)$ authentication data per node.

\paragraph{(ii)~Polynomial / homomorphic commitments.}
The second family of solutions doesn't only authenticate fragments, it also 
binds the fragments to $M$.
Two principal schemes belong to this family:
\begin{itemize}
  \item \textbf{Pedersen commitments}~\cite{Pedersen} have the form
    $c = g^{m} h^{r}$ in a prime-order group; they are
    computationally binding, and additively homomorphic. Pedersen commitments also provide the hiding property, which is not needed in VID. The additive
    homomorphism lets a verifier combine per-fragment commitments exactly as
    the linear erasure code combines fragments, so a fragment can be checked
    against an aggregate commitment to $M$. Pedersen needs \emph{no trusted
    setup}: its parameters are two generators with unknown relative discrete
    logarithm, which can be sampled transparently in any prime-order group.
    It is, however, not succinct: a single commitment binds one value, so
    authenticating a coded message uses a vector of $O(N)$ commitments, and
    checking a fragment is an $O(N)$ multi-exponentiation over that vector.
  \item \textbf{KZG polynomial commitments}~\cite{KZG} treat the encoded
    fragments as evaluations $\phi(1)$, $\ldots$, $\phi(N)$ of a single low-degree
    polynomial $\phi$. The sender commits to $\phi$ with one constant-size
    group element, and each fragment ships with a constant-size opening proof
    verified by a single pairing check. Both are independent of the number of
    fragments, unlike a Merkle root's $\Theta(\log N)$ authentication path, and
    the commitment is additively homomorphic, so a commitment to a linear
    combination of fragments can be derived from the underlying commitments.
    This succinctness rests on a \emph{trusted setup}: the structured reference
    string (SRS) is generated from secret randomness that must be destroyed,
    and should it leak, binding fails. KZG further requires pairing-friendly
    groups and more expensive proof generation, and the trusted-setup
    assumption is undesirable or unavailable in some deployments, motivating
    transparent alternatives.
\end{itemize}

\paragraph{Up-front prevention.}
A homomorphic commitment or fingerprint lets a verifier reject any off-codeword
fragment on receipt. A linear-sized one --- the Pedersen commitment vector, or
the homomorphic fingerprints of AVID-FP~\cite{HGR} --- binds the whole codeword
and prevents equivocation outright, at $O(N)$ per fragment. Constant-sized KZG
commitments require additional machinery to match this guarantee: prior work
either adds an explicit degree proof~\cite{KZGSecurityVSS} or restructures the
commitment with column commitments and row-wise encoding~\cite{SemiAVIDPR}, both
of which reintroduce an $O(N)$ vector.

\subsection{Existing AVID Solutions}
\label{sec:related-avid}


\subsubsection{Fingerprint-Based AVID (AVID-FP)}
In the original AVID, the sender sends each node the vector
$\big(H(f_1), \ldots, H(f_N)\big)$ of hashes of all $N$ fragments, so the receiver can verify any other fragments it may receive during reconstruction. 
Hendricks, et al. (AVID-FP)~\cite{HGR} optimize the original
AVID protocol by replacing the plain hashes with \textbf{homomorphic
fingerprints}. A homomorphic fingerprint can be thought of as similar to a homomorphic commitment, except it does not have the hiding property. The sender erasure-codes $M$ and computes a fingerprinted
cross-checksum, sending each receiver its fragment together with this
checksum. Because the fingerprints preserve the algebraic constraints of the
erasure code, each receiver can independently verify that its fragment
corresponds to the original data block. This removes the need for receivers
to echo their large, full fragments to one another; instead, receivers reach
agreement by reliably broadcasting only the succinct fingerprinted
cross-checksum.
This approach prevents equivocation upfront. The fingerprints \emph{are} the
homomorphic openings, so binding and corrupt-sender prevention coincide and
neither detect-and-reject nor a fragment-level reliable broadcast is needed.
AVID-FP's contribution is therefore largely on the coordination axis as nodes agree on the succinct checksum rather than echoing full fragments. \DM{you did not discuss the EC parameters, storage blow-up, and pruning here., but you did below. Why?}

\subsubsection{Polynomial-Commitment AVID}
A more recent line~\cite{AlhaddadDVZ21} uses polynomial commitments (e.g.,
KZG) in place of homomorphic fingerprints to compress per-fragment
authentication data from $O(N)$ to $O(1)$. The sender commits
to a polynomial whose evaluations are the fragments, distributes per-fragment
openings, and each receiver verifies its share against the single
commitment. With the standard $k = F+1$ parameterization, each node still
stores a fragment of size $|M|/(F+1)$, so the dispersal-time storage blowup
is $\frac{N}{F+1} \approx 3$ like the hash- and Merkle-based schemes; what
these schemes gain is $O(1)$ authentication data and communication blowup
that asymptotically matches lower bounds~\cite{LDDRVXG21}. Variants
that re-encode for long-term storage can approach the optimal blowup
$\frac{N}{N-F}$. The trade-off is a heavier cryptographic setup --- trusted
or transparent SRS, pairing-based verification, and more expensive proof
generation --- and, as with the schemes above, a static post-dispersal
storage profile: the per-node fragment size is fixed at dispersal time and
does not adapt as additional honest nodes acknowledge.

\subsubsection{Dispersal in Production BFT}
Several recent BFT protocols integrate erasure-coded dispersal with consensus
as a way to amortize the cost of broadcasting large transaction batches.
DispersedLedger~\cite{DispersedLedger} decouples consensus on a small header
from later retrieval of the full payload, allowing the payload to remain
dispersed across replicas until needed. The dispersal layer uses a fixed
$(k, N)$ AVID instantiation with $k = F+1$; nodes store a single fragment of
size $|M|/(F+1)$ and the storage profile is invariant to network conditions.
Its dispersal layer is a Merkle-based AVID paired with the detect-and-reject remedy.

Kudzu~\cite{Kudzu} pursues a different point in the design space: a
leader-based BFT atomic broadcast protocol that integrates erasure-coded
dispersal directly into its voting path. The leader sends erasure-coded
fragments of each proposed block to all other replicas, which in turn
re-broadcast their fragments together with first-round votes; by combining
fragments received from a sufficient number of replicas, every replica can
reconstruct and finalize the block in just $2\delta$ network delays, where
$\delta$ bounds the network delay, on the fast path. Kudzu uses a parameterized resilience regime $N = 3F + 2P + 1$
(where $F$ is the worst-case Byzantine threshold and $P$ controls the
fast-path quorum) and chooses its encoding parameters to match worst-case
dispersal needs across this regime; as with DispersedLedger, the per-replica
storage footprint is fixed at dispersal time and does not adapt to runtime
network behavior.
Kudzu's distinguishing choice is on the coordination
axis: rather than a standalone reliable broadcast or signature gather,
dispersal is folded into the consensus voting path.

\medskip
The common thread across these schemes is that the per-node storage profile
is fixed at dispersal time. AVID-RBC is the partial exception: it does drive
storage down toward the optimal $\frac{N}{N-F}$, but only by encoding the
message a second time after it has been committed and re-distributing the
fragments. This requires a coordinated re-encode, and isn't something
individual nodes can perform autonomously. 
\section{Deferred Proofs}
\label{app:deferred-proofs}

\begin{lemma}[Root agreement, full proof]
No two correct parties adopt distinct roots for the same
$\mathit{id}$.
\end{lemma}
\begin{proof}
A correct party adopts a root by one of two rules: an echo quorum of
$2F{+}1$ accepted \textsc{Echo}s, or $F{+}1$ \textsc{Done}s over a common root.

\emph{Echo-path adoptions.}
A correct party adopts $r$
by this rule only after collecting
$2F+1$ accepted \textsc{Echo}
messages for $r$ from distinct
parties. A correct party records at most one root per source, so no party contributes to the tallies of two roots at
the same adopter. Suppose correct parties adopt $r$ and $r' \neq r$. The two
echo sets each have size $2F+1$; on
$N = 3F+1$ parties they
intersect in at least $2(2F+1) - N \geq F+1$ parties, hence
in at least one correct party. A correct party sets $\mathit{echoed}$ once and
broadcasts \textsc{Echo} for a single root, so it cannot have echoed both $r$ and $r'$. Thus all correct parties that adopt a root
by this rule adopt the same
$r^\star$.

\emph{\textsc{Done}-path adoptions.} We show by induction on the order in
which correct parties broadcast \textsc{Done} that every correct \textsc{Done}
is over $r^\star$. Consider the first correct party to broadcast \textsc{Done}
for $\mathit{id}$. It adopted via the echo path and therefore announces $r^\star$. Now consider any
later correct broadcaster. It adopted either via the echo path, giving
$r^\star$, or on $F{+}1$ \textsc{Done}s over a common root. The latter set
contains at least one \textsc{Done} from a correct party, which by the induction
hypothesis is over $r^\star$, so the adopted root is $r^\star$ and the party
announces $r^\star$. Hence every correct \textsc{Done} is over $r^\star$, and a
\textsc{Done}-path adoption never yields any other root.
\end{proof}

\begin{lemma}[Totality, full proof]
If any correct
party outputs $\textsc{Completed}(M^\star)$ for $\mathit{id}$, then every correct party
eventually outputs $\textsc{Completed}(M^\star)$. If any correct party outputs
$\textsc{Abort}$, then no correct party outputs $\textsc{Completed}$, and every correct
party that adopts $r^\star$ outputs $\textsc{Abort}$.
\end{lemma}
\begin{proof}
A correct party outputs only after adopting a root; let $r^\star$ be the root it
adopted. By Lemma~\ref{lem:root-agreement} no correct party adopts a different
root, so $r^\star$ is the unique root adopted by correct parties.
By Lemma~\ref{lem:binding}, \Call{IsCodeword}{} evaluated against
$r^\star$ is a fixed value, so exactly one of the two cases below holds.

\textbf{Case A: $r^\star$ commits a codeword.}

\textbf{Step 1 (every correct party adopts $r^\star$).} The outputting party
counted $2F{+}1$ \textsc{Done}s over $r^\star$ from distinct senders, of which
at most $F$ are Byzantine, so at least $F{+}1$ are from correct parties; call
this set $Q$. Each member of $Q$ broadcast its \textsc{Done} to all parties, and
messages between correct parties are eventually received, so every correct party
eventually holds $\ge F{+}1$ \textsc{Done}s over $r^\star$ in
$\mathit{Dn}[\cdot]$. A correct party that has not already adopted therefore
adopts by the amplification rule, and by Lemma~\ref{lem:root-agreement} the root
it adopts is $r^\star$.
If the outputting party instead completed on the fast path, it
accepted \textsc{Echo}s over $r^\star$ from all $N$ parties, at least $2F{+}1$
of them correct; these were broadcast to all parties, so every correct party
eventually collects an echo quorum for $r^\star$ and adopts it by the echo
rule.

\textbf{Step 2 (every correct party obtains $2F{+}1$ fragments).}
By the induction in Lemma~\ref{lem:root-agreement}, the first correct party
to announce \textsc{Done} for $\mathit{id}$ adopted $r^\star$ via an echo quorum
of $2F{+}1$ accepted \textsc{Echo}s. At most $F$ of those echoers are Byzantine,
so at least $F{+}1$ are correct; call this set $Q'$.
Members of $Q'$ broadcast their \textsc{Echo}s
to all parties, so every correct party eventually receives $\ge F{+}1$ accepted
\textsc{Echo}s under $r^\star$. Each carries its sender's assigned pair and the
pairs partition $\{1,\ldots,2N\}$, so these supply $\ge 2(F{+}1) = 2F{+}2 \ge
2F{+}1$ leaves at distinct indices, all verified under $r^\star$. A party that
adopted on \textsc{Done}s may hold fewer than $2F{+}1$ fragments at the instant
it adopts.
\Call{TryFinishFrom}{} is re-invoked on every \textsc{Echo} that extends
$B$ under the adopted root, so it eventually runs at that party with
$\ge 2F{+}1$ distinct fragments.

\textbf{Step 3 (every correct party outputs).} Every correct party finds \Call{IsCodeword}{}
true, stores its pair, and broadcasts \textsc{Done}
with value $r^\star$. There are $\geq 2F+1$ correct
parties, so $\geq 2F+1$ correct
\textsc{Done}s over $r^\star$ are eventually received by
every correct
party, and each outputs $\textsc{Completed}(M^\star)$.

\textbf{Case B: $r^\star$ commits a non-codeword.} We show no
correct party outputs \textsc{Completed}, and every correct party that adopts
$r^\star$ outputs \textsc{Abort}.
Since $r^\star$ is a non-codeword, no $2F{+}1$ leaves verified under $r^\star$
re-encode to $r^\star$, so \Call{IsCodeword}{} is false at every correct party
that reconstructs; hence no correct party broadcasts \textsc{Done}, and no correct party outputs
\textsc{Completed}, on the slow path for want of a \textsc{Done}
quorum and on the fast path for want of $\mathit{doneSent}$. It follows that no correct party adopts by the amplification
rule either, since that rule requires $F{+}1$ \textsc{Done}s over a common
root. Every correct adopter is therefore an echo-path adopter, and adoption by
that rule requires $2F{+}1$ accepted \textsc{Echo}s carrying disjoint full
pairs, so such a party holds $\ge 2F{+}1$ distinct fragments verified under
$r^\star$; it runs \Call{TryFinishFrom}{}, finds \Call{IsCodeword}{} false, and
outputs \textsc{Abort}.
No correct party prunes either, since $|D| = N$ would require $2F{+}1$
correct parties to have announced \textsc{Done} over $r^\star$.
\end{proof}

\begin{lemma}[Pruning safety, full proof]
If a correct party prunes a fragment for
$\mathit{id}$, then
$M^\star$ remains reconstructable
from the fragments held by correct parties.
\end{lemma}
\begin{proof}
A correct party prunes only when $|D| = N$ and $\mathit{doneSent}$ holds.
The latter means the pruning party itself found \Call{IsCodeword}{} true against
$r^\star$, so $r^\star$ commits a codeword. The former means it
received $\langle \textsc{Done}, \mathit{id}, r^\star \rangle$ from all $N$
parties. Of
these $N$ senders at most $F$ are Byzantine, so at least
$2F+1$ are correct. By
channel authenticity each of those announcements was genuinely sent by the
correct party it is attributed to, so each of these $\geq 2F{+}1$ correct
parties executed the
\textsc{Done} broadcast and had stored its
assigned pair $S_k$ beforehand.

Pruning is performed independently by each party
as its own $|D|$ reaches $N$, so at any point the
system may be in a mixed
state: some correct parties have pruned and others have not,
and each may prune at a different time, and some may never prune at all.
We show reconstructability holds in
every such state.
Fragments are
distinct across parties by the disjoint partition, and pruning discards at most
one fragment per party, so each of these $\geq 2F+1$ correct parties retains at
least one verified fragment at a distinct index --- regardless of which parties
have pruned. These $\geq 2F+1$ distinct verified fragments meet the threshold of
the $(2F{+}1,2N)$ code, so $M^\star$ remains reconstructable (Case~A of
Lemma~\ref{lem:binding}).
\end{proof}

\begin{lemma}[Retrieval termination, full proof]
If dispersal for $\mathit{id}$
completed with a $\textsc{Completed}$ output, then every
correct party that invokes \textsc{Retrieve} eventually outputs $M^\star$.
\end{lemma}
\begin{proof}
A $\textsc{Completed}$ output required $2F+1$ $\textsc{Done}$
announcements over
$r^\star$, at least $F+1$ from correct parties, each of which stored its assigned
pair before broadcasting (Lemma~\ref{lem:dispersal-safety}).
On the fast path it instead required \textsc{Echo}s over
$r^\star$ from all $N$ parties, at least $2F{+}1$ from correct parties, each
of which stored a verified fragment of its own pair before echoing and retains
at least one thereafter (Lemma~\ref{lem:dispersal-safety}); these $\geq 2F{+}1$
distinct fragments suffice regardless of pruning, so the two cases below
concern the slow path.

\emph{No correct party has pruned.} The $\ge F{+}1$ correct announcers
still hold their full pairs, i.e.\ $\ge 2(F{+}1) = 2F{+}2$ verified fragments;
by disjointness of the partition these indices are distinct.

\emph{Some correct party has pruned.} That party pruned only on $|D| = N$
with $\mathit{doneSent}$ set, so by the counting in
Lemma~\ref{lem:pruning-safety} at least $2F{+}1$ correct parties had stored
their assigned pairs before announcing, and each retains at least one verified
fragment at a distinct index after pruning. This bound is independent of which
of them have pruned so far, so it holds in every mixed state.

In every case the correct parties jointly hold $\ge 2F+1$ distinct fragments
verified under $r^\star$. Every such holder $P_i$ holds at least one fragment that verifies under
$r^\star$, so on receiving $R$'s $\textsc{Get}$ it replies with
$\textsc{GetResp}$, whether or not it has yet adopted
$r^\star$. $R$ accumulates every
correct-held fragment into $B_R$, and since these number $\ge 2F+1$ at distinct
indices, $R$'s threshold $|\{i : (i,\cdot,\cdot) \in B_R\}| \ge 2F+1$ eventually
holds. Since dispersal output $\textsc{Completed}$, $r^\star$ commits a codeword
(Case~A of Lemma~\ref{lem:binding}) and $R$ outputs its decoding $M^\star$.
\end{proof} 

\end{document}